\documentclass[submit]{aiaa-tc}

 \usepackage{varioref}
 \usepackage{wrapfig}
 \usepackage{threeparttable}
 \usepackage{dcolumn}
  \newcolumntype{d}{D{.}{.}{-1}}
 \usepackage{nomencl}
  \makeglossary
 \usepackage{subfigure}
 \usepackage{subfigmat}
 \usepackage{fancyvrb}
 \fvset{fontsize=\footnotesize,xleftmargin=2em}
 \usepackage{lettrine}
 \usepackage[dvips]{dropping}
 \usepackage[colorlinks]{hyperref}
\usepackage{tabu}

 \usepackage{amsmath} 
 \usepackage{amsthm}

\usepackage{amsfonts}
\usepackage[english]{babel}
 \usepackage{amsfonts} 
\usepackage{enumerate}
\usepackage{graphicx}
\usepackage{caption}
\usepackage{setspace}
\usepackage{algorithmic}

\usepackage[usenames]{color}
\usepackage[numbers]{natbib}

\newcommand{\HL}[1]{{\color{blue} \bf XX #1 XX\ }}

 \title{Distributed Unknown-Input-Observers for Cyber Attack Detection and Isolation in Formation Flying UAVs }

 \author{  Lebsework Negash\thanks{PhD Candidate, Department of  Aerospace Engineering, KAIST. 291 Daehak-ro, N7-3332, Yuseong, Daejeon 34141, Republic of Korea; {\tt lebsework@kaist.ac.kr}},~
     Sang-Hyeon Kim\thanks{PhD Candidate, Department of  Aerospace Engineering, KAIST. 291 Daehak-ro, N7-3333, Yuseong, Daejeon 34141, Republic of Korea; {\tt k3special@kaist.ac.kr}},
  and Han-Lim Choi\thanks{Associate Professor, Department of Aerospace Engineer, KAIST. 291 Daehak-ro, N7-4303, Yuseong, Daejeon, 34141, Republic of Korea; Tel: +82-42-350-3727; Fax: +82-42-350-3710; {\tt hanlimc@kaist.ac.kr}. AIAA Member. Corresponding Author.}
 }
 
 \AIAApapernumber{YEAR-NUMBER}
 \AIAAconference{Conference Name, Date, and Location}
 \AIAAcopyright{\AIAAcopyrightD{YEAR}}

\newtheorem{theorem}{Theorem}
 \newtheorem{corollary}{Corollary}
 \theoremstyle{definition}
  \newtheorem{definition}{Definition}

\begin{document}

\maketitle

\begin{abstract}
In this paper, cyber attack detection and isolation is studied on a network of UAVs in a formation flying setup. As the UAVs communicate to reach consensus on their states while making the formation, the communication network among the UAVs makes them vulnerable to a potential attack from malicious adversaries. Two types of attacks pertinent to a network of UAVs have been considered: a node attack on the UAVs and a deception attack on the communication between the UAVs. UAVs formation control presented using a consensus algorithm to reach a pre-specified formation. A node and a communication path deception cyber attacks on the UAV's network are considered with their respective models in the formation setup. For these cyber attacks detection, a bank of Unknown Input Observer (UIO) based distributed fault detection scheme proposed to detect and identify the compromised UAV in the formation. A rule based on the residuals generated using the bank of UIOs are used to detect attacks and identify the compromised UAV in the formation. Further, an algorithm developed to remove the faulty UAV from the network once an attack detected and the compromised UAV isolated while maintaining the formation flight with a missing UAV node.
\end{abstract}

\section{Introduction} \label{sec:intro}

Recent advancement in UAV’s capabilities, operating in an autonomous mode, a high-end computing and communication infrastructures onboard have spurred a wide interest from law enforcement to commercial sectors in deploying a large number of UAVs. These make the UAVs an ideal candidate for a coordinated task where it is not  possible to perform the task single-handedly and efficiently~\cite{miller2006mini,goodrich2008supporting,doherty2007uav} such as finding and rescue mission, or law enforcement such as border patrolling and drug trafficking monitoring~ \cite{waharte2010supporting,rudol2008human,matveev2011method,HomelandSecurity2014,bolkcom2004homeland,haddal2010homeland}.  These tasks rely on the  cooperative  control nature of multi-UAVs and their interaction with the environment with all its uncertainties. Flocking in birds or school of fish has a motion with a well-coordinated pattern. These inspired the robotic community to develop a similar structure for a coordinated task or flight formation control, which  basically uses a distributed control strategy.

The decentralized control nature of the system makes it vulnerable to malfunction and possible threats or attacks~\cite{MengGuo2012}. In addition, UAVs are  cyber-physical systems (CPS) with a tight integration of physical process, computational resource, measurement and communication capabilities. The control unit monitors and controls the system status while coordinating the flight through sensors and actuators onboard~\cite{pasqualetti2013attack}. Since many UAVs use off-the-shelf communication equipment and computing components (flight controller boards) with well-known protocols, these make them more open  and prone to cyber-attacks from adversaries~\cite{cardenas2008secure,kim2012cyber}. Thus, the use of standard protocols on mission critical systems leads to a source of cybersecurity vulnerability where adversaries are capable of exploiting commonly known Internet vulnerabilities~\cite{mander2007data,rodday2016exploring}. Many unmanned vehicles use encryption of data channels to prevent cyber attacks but relying on it as the only defense mechanism is misguided~\cite{nilsson2009defense}. In addition, there are attacks on multiple sensors which can corrupt the state of the UAV without the need of breaking the encryption. Examples of these attacks are spoofing of GPS or Automatic Dependent Surveillance-Broadcast (ADS-B) signals ~\cite{warner2003gps,krozel2005independent,kim2012cyber}. Traditional computer science cybersecurity approach focused on the integrity of data, encryption and restricting access to sensitive data. It  only protects the computing resource of the cyber-physical system and is oblivious to the cyber-physical interaction \cite{kwon2015real}. While this is necessary to keep the system secured, it fails to restrict access to the system~\cite{shull2013analysis,banerjee2012ensuring}.  These put the security of the previously closed and isolated control systems to a new dimension of threats where classical control cannot deal with in its fault identification and isolation (FDI) scheme used to identify and clear faults in control systems. Thus, the problem of the UAV security is studied here from  a complementary control theory and fault detection perspective where compatibility of the measurements with the underlying physical process of the control mechanism is exploited.

A formation flight of UAVs is an example of multi-agent systems performing a shared task using inter-vehicle communication to coordinate their action and reach a consensus on the desired moving formation setup. As defined by Olfati-Saber et al~\cite{olfati2007consensus}, ``\textit{consensus means to reach an agreement regarding a certain quality of interest that depends on the state of all agents. Consensuses algorithm is an interaction rule that specified the information exchange between an agent and all of its neighbors on the network}." The consensus problem for networks of dynamic systems  was presented by Olfati-Sabri and Murray in\cite{murray2003consensus} and showed that connectivity of the network is the key to reaching a consensus. Fax and Murray~\cite{fax2004information}  presented a vehicle cooperative network performing a shared task. They made use of Nyquist criterion which uses the graph Laplacian eigenvalues
to prove the stability of the formation. A decentralized control of vehicle formation was  presented by Lafferriere and Williams~\cite{lafferriere2005decentralized} using a consensus algorithm, where they proved the necessary and sufficient condition for the existence of a decentralized linear feedback controller.

Much of the recent research on distributed control system security in cyber-physical systems focused on electric power gird estate estimator and sensors anomaly affected by adversaries manipulating sensor measurements. Observer-based approaches have been studied for a networked power system fault detection ~\cite{aldeen2006observer,scholtz2008graphical}. An intrusion detection scheme for linear consensus network with a  misbehaving node was presented in~\cite{pasqualetti2007distributed}, where the authors used unknown input observer (UIO). Other results were presented in \cite{AndreTeixeria2010,shames2010distributed} where the authors used a bank of UIO systems fault and cyber attack detection for a network of a power system. The sufficient condition for the existence of a bank of UIOs was given as  the graph representation of the system being connected. The nodes in the network were modeled as a second-order linear time-invariant system.

A distributed real-time fault detection in a cooperative multi-agent system was presented in\cite{guo2012distributed}. The authors introduced a fault detection framework in which each node monitors its neighbors using local information. The authors in \cite{barua2011hierarchical},  introduced a fuzzy rule-based hierarchical fault detection and isolation framework  for spacecraft formation. Simple fuzzy rules were developed to describe the relationship between faults and to isolate faulty satellites in the formation. In another work \cite{daigle2007distributed}, a distributed, model-based and qualitative fault diagnosis approach for formations of mobile robots was presented. The model of the mobile robot and the communication among them were modeled as a bond graph. The authors in ~\cite{meskin2006fault}, investigated a geometric distributed FDI methodology by developing a bank of local/decentralized detection filters for detecting faults in other spacecraft, while they are in a formation flight by determining the required observability subspace of the local system. A relationship between the number of detectable malicious or faulty  nodes and the topology was investigated in \cite{sundaram2011distributed}. The authors showed that the topology of the network completely characterized the resilience of the linear iterative system. It showed that for $f$ malicious nodes, a node was able to detect the faults of all nodes if the node has at least $2f$ vertex-disjoint paths from every other non-neighboring nodes. Authors in \cite{kwon2014analysis} considered a cyber attack on the critical part of unmanned aerial system, the state estimator. They showed how a stealthy attack can fail the estate estimator without being detected by the monitoring system.

Control theory and fault detection schemes can be used in a distributed system setup to detect a malicious cyber attack on a network of UAVs in a formation flight. The main contribution will be the detection of a possible cyber attack on a network of UAVs in a formation flying setup using a bank of UIO observers. In addition to the detection of a possible cyber attack, the UIO will be used to identify the compromised UAV in the network.
The cyber attack is  modeled as a node attack and a deception attack on the communication channel while the UAVs are performing the coordinated task. Furthermore, a faulty UAV removing algorithm developed to remove the malicious or under attack UAV  from the network. The algorithm will remove the faulty UAV for a $2$-connected network while maintaining connectivity of the flying formation network and the formation setup with a missing node UAV. The key contributions of this work are a distributed UIO based cyber attack detection and a safe removal of compromised UAV from a formation flying network. While preliminary ideas and results were reported in the authors' earlier work  \cite{negash2016unknown}, this paper presents expanded theoretical results including attack isolation as well as a new set of numerical results.

The structure of the rest of the paper is as follows. In section II, the dynamics of UAVs and their communication model in the formation control is presented. In section III, a formal definition and model of a node and a communication path deception cyber attack in the formation are described. UIO based fault detection and  a compromised UAV identification with a faulty UAV removal algorithm  are presented. A simulation result of a formation control, cyber attack detection and removal of compromised UAV in the formation is presented in section IV.  Finally, in section V, a summary of the main result and some thoughts in the future direction are provided.

\section{Formation Control of UAVs} \label{sec:formation}
In this section, the formation control adapted from \cite{fax2004information} so that it will suit for the specific purpose of this paper and dynamics of the UAVs considered. Starting from here, in a formation flying the UAVs will be referred as agents or nodes interchangeably.

Consider $N$ UAVs coordinating themselves to achieve a pre-specified formation defined by relative positions with respect to each other. To describe the interaction architecture in a formal manner, consider an undirected graph $\mathcal{G}=(\mathcal{V},\mathcal{E})$, where $\mathcal{V}$
the set of nodes $\mathcal{V}=\{{v_1,...v_N}\}$   and $\mathcal{E}$ set of edge, $\mathcal{E}\subseteq\mathcal{V}\times\mathcal{V}$. The neighbors of UAV $i$ are denoted by $\mathcal{N}_i=\{j \in\mathcal{V}:(i,j)\in\mathcal{E}\}.$ Every
node represents a UAV and the edges correspond to the inter-vehicle communication.
The adjacent matrix $\mathcal{A} \in \{0,1\}^{N \times N}$ represents the adjacency relationship in the graph $\mathcal{G}$ with an element $a_{ij}=1$ if $(v_i,v_j )\in\mathcal{E}$  and $a_{ij}=0$ otherwise. The neighbor of agent $i$, denoted as $\mathcal{N}_i$ is the set of agents such that $a_{ji} = 1$. The graph Laplacian is defined as
 \begin{equation}
L_{\mathcal{G}} = \mathcal{D} - \mathcal{A}\label{eq:L_g}
\end{equation}
where $\mathcal{D}$ is a diagonal matrix with $d_{ii} $ representing the cardinality of $ \mathcal{N}_i$.

Each UAV's motion in $d$-dimensional Euclidean space is modeled as a second order system:
 \begin{equation}
 \dot{x}_i =A_i x_i+B_i u_i, \qquad x_i \in\mathbb{R}^{n} \label{eq:x_i}
 \end{equation}
where the state variable $x_i$ consists of the configuration variables (i.e., position) and their derivatives (i.e., velocity) and the control input $u_i$ represents the acceleration commands; the system matrices take the form of:
$$
A_i = {\tt diag} \left\{ \begin{bmatrix} 0 & 1 \\ \alpha_{ij} & \beta_{ij} \end{bmatrix} \right\},~~j = 1, \dots, d , \qquad B_i = I_n \otimes \begin{bmatrix} 0  \\  1 \end{bmatrix}
$$
with appropriate $\alpha_{ij} >0, \beta_{ij} \geq 0$, where $I_n$ is $n\times n$ identity matrix and $\otimes$ denotes the Kronecker product.
\iftrue
A UAV is assumed to obtain information about its motions via the local observation model:
\begin{equation}
y_i = C_i x_i. \label{eq:local_sense}
\end{equation}
For notational simplicity, the derivation hereafter considers the case where a UAV can access to its full local state, i.e., $C_i = I_{n}$, while Remark \label{rem:ci} will discuss on how the main results can extend to a more generic case.
\fi

	\begin{figure}[t]
		\centering
		\includegraphics[width=0.5\textwidth]{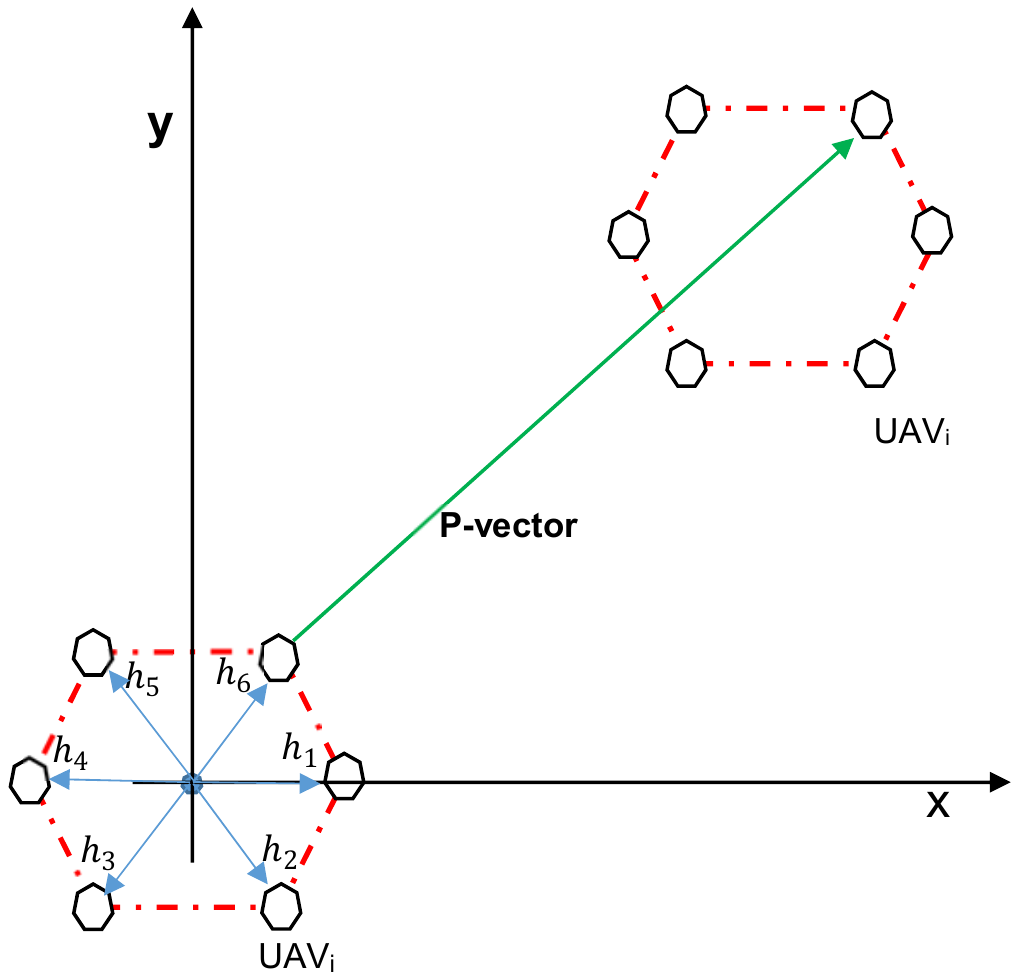}
		\caption{UAVs in a hexagonal formation for $N$ = 6.}
		\label{fig1}
	\end{figure}

Let denote $\tilde{h}_i$ as some possible desired position vector of agent $i$ for $i \in \{1, \dots N \}$. For example, if the desired formation is a planar hexagon (as in Figure \ref{fig1}), one choice of $h_i$ is:
$$
\tilde{h}_i = \begin{bmatrix} R \cos ( \pi/3 \times i ) & R \sin (\pi/3 \times i )  \end{bmatrix}^T.
$$
The UAVs are called in formation $h$ at time $t$ if there are constant vector $p$ such that
$$
x_{i,pos} - \tilde{h}_i = p, \qquad \forall i \in \{1, \dots, N\},
$$
where $x_{i,pos}$ is the $n$-dimensional vector consisting of the odd entries of $x_i$. Also, the UAVs are said to converge to the formation $h$ when the limit $\lim_{t \rightarrow \infty} x_{i,pos}(t) - \tilde{h}_i $ exists and are same for all $i \in \{ 1, \dots, N \}$~\cite{lafferriere2005decentralized}.

 Interaction of agent $i$ with other agents for formation flight is through the control input term $u_i$. To achieve formation, the feedback signal used to generate this control input is the difference between its own offset from the desired formation vector and that of the neighboring agents:
\begin{equation}
u_i = \frac{K_i}{|\mathcal{N}_i | }  \sum_{j \in \mathcal{N}_i} \left[ (x_i - h_i ) - (x_j - h_j) \right], \qquad \forall i \in \{1, \dots, N \} \label{eq:u_i}
\end{equation}
where $h_i =  \tilde{h}_i \otimes [1~0]^T$, with some feedback gain $K_i \in \mathbb{R}^{n \times 2n} $. The cardinality $|\mathcal{N}_i |$ is used for normalization purpose~\cite{fax2004information,lafferriere2005decentralized}. Since UAV motion is modeled as a second-order system with acceleration input, $K_i$ takes the form of
$$
K_i = I_n \otimes \left[ k_{i, pos}, k_{i, vel} \right].
$$
With the state equation in (\ref{eq:x_i}) and control input in (\ref{eq:u_i}), the overall closed-loop dynamics of the fleet can be written as:
\begin{equation}
\dot{x} = A x + B K L ( x- h) \label{OVR_allFormation}
\end{equation}
with the overall state $x = [x_1^T, \dots, x_N^T]^T$ and desired formation $h = [h_1^T, \dots, h_N^T]^T$, where
$$
A = I_N \otimes A_i, ~B = I_N \otimes B_i,~K = I_N \otimes K_i, ~L = L_{\mathcal{G}} \otimes I_n.
$$

\begin{theorem}[Theorem of \cite{fax2004information}]
A  controller $K$ stabilizes the formation dynamics in (\ref{OVR_allFormation}) if and only if it simultaneously stabilizes the individual $N$-UAV systems.
\begin{equation}
\begin{split}
\dot{x_i}&=A_{i}x_i+B_{i} u_i, ~~~~~  i=1,...N, \\
y_i &=C_i x_i
\end{split}
\end{equation}
\end{theorem}	
\begin{proof}
	Let $M$ be a Schur transformation matrix of $L_\mathcal{G}$  where $\tilde{L}_{\mathcal{G}} =M^{-1} L_\mathcal{G} M$ is upper triangular\cite{horn2012matrix}.
	The diagonal entries of $\tilde{L}_{\mathcal{G}}  $ are the eigenvalues of $L_{\mathcal{G}} $. Clearly $M\otimes I_{n}$ transforms $L_{\mathcal{G}}\otimes I_{n}$ into $\tilde{L}_{\mathcal{G}}  \otimes I_{n}$.
	Calculating directly,
	\begin{equation}
	(M^{-1}\otimes I_{n})(A+BKL)(M\otimes I_{n})=I_N\otimes A_{i}+\tilde{L}_{\mathcal{G}}\otimes B_{i}K_{i}
	\end{equation}
	The right-hand-side is an upper triangular and the $N$ diagonal subsystems are of the form:
	\begin{equation}
	A_{i}+\lambda_i B_{i}K_{i}
	\end{equation}
	where $\lambda_i$  is an eigenvalue of $L_{\mathcal{G}}$. There is one block for each eigenvalue of the Laplacian. Therefore, the eigenvalues of  $A+BKL$ are those of $A_{i}+\lambda_i B_{i}K_{i}$ where $\lambda_i$ is the eigenvalue of $L_{\mathcal{G}}$ corresponding to UAV $i$.	Thus, the stability analysis of the $N$ formation UAVs can be achieved by analyzing the stability of a single $UAV_i$ with the same dynamics modified by scaler representing the interconnection Lplacian eigenvalues. Consequently, designing the feedback gain $K_{i}$, stabilizing the single vehicle, scaled by the eigenvalue of the Laplacian  leads to a stable formation.
\end{proof}

\section{Distributed Cyber Attack Detection and Isolation} \label{sec:detection}

In this section, the main methodology for detection and isolation of cyber attacks over a network of UAVs in formation control is presented.  The method takes advantage of fault detection and isolation (FDI) schemes for handling sensor-actuator faults of dynamic systems. The effect of cyber attack on a networked cyber-physical systems is in essence the inability of a certain component of over the network; thus, can be treated as a fault in the input/ouput elements of the overall system~\cite{amin2013cyber,pasqualetti2012attack}.  A specific model-based diagnostic scheme, called unknown input observer, is considered to generate the residual signals for checking the presence of faults. The UIO-based diagnostic scheme is used to detect a class of adversarial scenarios based on a generalized fault model. The attacks dealt herein are communication network-induced deception attack and node attack.

\subsection{Unknown Input Observer}
In model-based FDI system, a residual, which is generated as the difference between the measurement and estimate of the states of the process, is used as an indicator of a presence of a fault. The residual should be close to zero if and only if a fault does not occur in the system. This section briefly summarizes the unknown input observer (UIO) scheme~\cite{Chen1999} for fault diagnosis of a linear system. .

The UIO considers a fault-free system in the form of the following linear time-invariant system:
\begin{equation}
\begin{split}
\dot{x}(t)&=Ax(t)+Bu(t)+Ed(t) \\
y(t)&=Cx(t)
\end{split} \label{eq:unknownInput}
\end{equation}
where $x\in\mathbb{R}^n$, $u\in\mathbb{R}^m$ are the state and known input vectors, respectively; $d\in\mathbb{R}^n$ is the unknown input vector and the associate input matrix $E$ is of full column rank. In the presence of a fault, the system dynamics is given by:
\begin{equation}
\begin{split}
\dot{x}(t)&=Ax(t)+Bu(t)+B_f f(t) +Ed(t) \\
y(t)&=Cx(t)
\end{split} \label{eq:unknownInputwithFault}
\end{equation}
where $f(t)$ is an unknown \textit{scalar} time-varying function representing evolution of the fault and/or attack. The fault distribution matrix $B_f$ is assumed to be of full column rank.

A UIO for the system in (\ref{eq:unknownInput}) is given as
\begin{equation}
\begin{split}
\dot{z}(t)&=Fz(t)+TBu(t)+Py(t)\\
\hat{x}(t)&=z(t)+Hy(t)
\end{split} \label{eq:obs}
\end{equation}
   where $\hat{x}$ is the estimated state vector and $z\in\mathbb{R}^n$ denotes the state variables of the observer. 	The matrices in the above observer equations must be designed in such a way that it can achieve decoupling from the unknown input and meet the stability requirement of the observer. To achieve these condition, choose the matrices $F,T,P$ and $H$ satisfying the following conditions
\begin{equation}
\begin{split}
	(HC-I)E&=0\\
	T&=(I-HC)\\
	F&=(A-HCA-P_1C)	\\
	P_2&=FH
\end{split}  \label{eq:obsCondsMatric}
\end{equation}
	where $P=P_1+P_2$.
	The state estimation error dynamics will then be:
	\begin{equation}
	\dot{e}=Fe(t)
	\end{equation}
	where $F$ is chosen so that all eigenvalues are stable, $e(t)$ will approach zero asymptotically,
	$\lim_{t \to+ \infty} e(t) = 0 $, regardless of the values of the unknown input $d(t)$.
\begin{definition}
	The state observer in (\ref{eq:obs}) is called an unknown input observer (UIO) if its state estimation error
	vector $e(t)=x(t)-\hat{x}(t)$ approaches zero asymptotically regardless of the unknown input $d(t)$.
\end{definition}
	\begin{theorem} \label{theo:2}
	\cite{Chen1999} The necessary and sufficient conditions for an UIO described
	by (\ref{eq:obs}) to be an observer for the system (\ref{eq:unknownInput}) are:
	\begin{enumerate}[a.]
		\item rank $(CE)=$ rank$(E)$
		\item $(C,A_1)$ is detectable pair where $A_1=A-HCA$
	\end{enumerate}
For the proof, reader advised to look into \cite{Chen1999}.
	\end{theorem}
	
\subsection{Cyber Attack on Formation}
For the network of UAVs in formation control using the method described in section \ref{sec:formation}, UIOs can be utilized as a mechanism to detect a possible cyber attack. Since there is no central agent who can gather all the state information of the UAVs, the UAVs should be able to detect cyber attacks relying only on the local communication with their neighbors. One way this work proposes to facilitate this distributed detection is for every UAV to have a \textit{bank} of UIOs each of which is associated with a particular attack origin. Then, they consequently coordinate a corrective action in the network.  Two attack types, node and communication deception attack, are modeled as an unknown disturbance in the UAV's dynamics. The bank of UIOs generates a structured set of residuals where each residual is decoupled from one and only one fault but sensitive to all other faults.

\subsubsection {Node attack}
Assuming the $k^{th}$ UAV is affected by an outside malicious agent and compromised as a unit. The UAV control input is corrupted and this UAV is no longer following the system-wide distributed control to perform the formation. This can be due to the incoming measurements of neighbors' states being affected by  DoS-type  attack \cite{teixeira2012attack,amin2009safe} or compromised at the signal receiver module of the UAV and making the UAV react to the compromised input. In this case, the attack affects the  $k^{th}$ node system state dynamics directly.

Such an attack is modeled as a disturbance to the system dynamics of node $k$:
\begin{equation}
\begin{split}
\dot{x}_k&=A_{k}x_k+B_{k}u_k +b_f^kf_k \label{eq:k nodeAtt dym}\\
y_k&=C_kx_k
\end{split}
\end{equation}
where $b_f^k\in\mathbb{R}^n$ is the  distribution vector and $f_k\in\mathbb{R}$ be the disturbance signal. The detection scheme employed on $i^{th}$ UAV, the global dynamics of the system can be described as (\ref{eq:k nodeFault}), and the system to be monitored at node $i$ with all possible faults in the formation is
\begin{align}
\dot{x}&=(A+BKL)x-BKLh+B_ff \label{eq:k nodeFault}\\
y_i&=\bar{C}_ix \nonumber
\end{align}
where $f=[f_1...f_N]^T$ and  $B_f\in\mathbb{R}^{Nn\times N}$ is a block diagonal matrix  in terms of ${b_f^i}$ elements.
Fault can be rewritten so that the effect of the fault in the $k^{th}$ UAV is evident 	
\begin{align}
	\dot{x}&=(A+BKL)x-BKLh+B_f^k f_k + B_f^{\bar{k}} f_{\bar{k}} \label{eq:k nodeIsoFault}\\
	y_i&=\bar{C}_ix \nonumber
	\end{align}
	where $B_f^k$ is the $k^{th}$ column of $B_f$, $f_k$ is the $k^{th}$ component of $f$, $B_f^{\bar{k}}$ is
	$B_f$ with the $k^{th}$ column deleted and $f_{\bar{k}}$ is the fault vector with the $k^{th}$ component removed.
	In order to make the observer insensitive to the $f_k$, this fault is regarded as an unknown input where $B_f^k$ is analogous to $E$ in (\ref{eq:unknownInput}). $B_f^k = [{b_f^k}^T ~ 0_{1  \times (N-1)n)}]^T$  where $b_f^k$ is an $n$ dimensional vector will all zero entries except one that corresponds a single fault of the faulty UAV $k$ and $\bar{C}_i = [C_i ~ 0_{(N-1)n \times (N-1)n}]$. The UIO implmented at the $i^{th}$ UAV, decoupled from $f_k$ and made insensitive to a disturbance in $k^{th}$ UAV, has the following dynamics:
	\begin{align}
	\dot{z}_i^k&=F_i^kz_i^k+T_i^kBu+P_i^ky_i \label{eq:k_UIO}\\
	\hat{x}_i^k&=z_i^k +H_i^ky_i \nonumber
	\end{align}
	with $\hat{x}_i^k\in\mathbb{R}^{nN}$ being the estimate of the $N$ UAVs' states insensitive to a fault in the $k^{th}$ UAV,
	which is computed by $i^{th}$ UAV. Since the formation network is only running the consensus algorithm, we consider the closed loop dynamics (\ref{eq:k nodeIsoFault}), with $u=0$ and hence, to incorporate in the observer design we take $B=0$.
	
	The above UIO exists if and only if it satisfies the Theorem \ref{theo:2}  conditions. The disturbance in the $k^{th}$  UAV (node) represented as unknown input by setting $E=B_f^k$, which is the $k^{th}$ column of $B_f$.

	\textbf{Detecting the Node attack }

	
	\begin{corollary} \label{corol:1}
		There exists a UIO for the system $((A+BKL),B_f^k, \mathcal{N}_i)$, if the following conditions are satisfied:

	\begin{enumerate}[a.]
		\item rank $(\bar{C}_iB_f^k)=$ rank$(B_f^k)$
		\item $(\bar{C},A_1)$ is detectable pair where $A_1=A-H \bar{C}A$
	\end{enumerate}
	\end{corollary}
\begin{proof}
	we have to show that  $$ rank (\bar{C}_iB_f^k)= rank(B_f^k)=1$$
	Denoting the row of $\bar{C}_i$ that reads the output of UAV $k$ be $\bar{c}_i^k$. It is obvious that $\bar{c}_i^k B_f^k =1$
	since $B_f^k$ is a vector with all entries are 0 except $k^{th}$ is 1. Therefore, with the same token  $\bar{C}_iB_f^k$
	is a vector with all entries are 0 except $k^{th}$ is 1, thus making its rank equal to 1.

	The second condition of corollary \ref{corol:1}, is equivalent to a condition that the transmission zeros from the unknown
	inputs to the measurement must be stable~\cite{Chen1999}, i.e. the matrix

$	rank
	\begin{bmatrix}
		sI_{nN}-(A+BKL)&B_f^k\label{eq:com_attModel}\\
	\bar{C}_i&0\\
	\end{bmatrix}
$
is of full rank for all $s$ such that $\mathfrak{Re}(s)\geq 0$, this can be proved as follows:\\
$$
rank	
\begin{bmatrix}
	sI_{nN}-(A+BKL)&B_f^k\\
		\bar{C}_i&0\\
	\end{bmatrix}
	=
	rank\begin{bmatrix}
	sI_{nN}-(A+BKL)\\
		\bar{C}_i\
	\end{bmatrix}
+
	rank  (B_f^k),
$$
	where a stable closed loop system $(A+BKL)$, is a full column rank. \\
Therefore,
$
rank	
\begin{bmatrix}
sI_{nN}-(A+BKL)&B_f^k\\
\bar{C}_i&0\\
\end{bmatrix}
= nN+1
$ 		
\end{proof}

	Once the existence of the UIO, from the system dynamics in (\ref{eq:k nodeIsoFault}) and observer dynamics from (\ref{eq:k_UIO}) are verified, it is easy to drive the error dynamics and the residual as
		\begin{align}
	\dot{e}_i^k&=F_i^k e_i^k-T_i^kB_f^{\bar{k}} f_{\bar{k}} \label{eq:k_errorDyn}\\
	\ r_i^k&=C_ie_i^k \nonumber
	\end{align}
	where $f_{\bar{k}}$ is obtained by removing the $k^{th}$ fault element of $f$. Note that the residual dynamics are driven by all except the $k^{th}$ fault if $T_iB_f^{\bar{k}} \neq 0$ for $ i \neq k$, making the residual sensitive to all but the $k^{th}$ fault.
	\\

		The bank of UIO observers at UAV $i$ generates residual signals for each of its neighbors $\mathcal{N}_i$.
	Since $B_f^{\bar{k}}$ has full column rank, the UIO  residual $r_i^k$ is insensitive only to $f_k$, treating it as unknown
	input. With this in mind, the following threshold $(T)$ logic can be set:
\begin{algorithmic}
		\IF{$ \|r_i^k  \|< T_{fk} $     ,$\forall k\in \mathcal{N}_i$ }
		\STATE No fault in the neighbor
		\ELSIF{$\| r_i^k \| < T_{fk} $    ,$\forall k \neq j \in \mathcal{N}_i$ and       $ \|r_i^j \|\geq T_{f_k}$    ,$\forall j\neq k  \neq j \in \mathcal{N}_i$}
		\STATE Fault in the neighbor node $k$
	
		\ENDIF
	\end{algorithmic}

\subsubsection{Attack on the outgoing communication of a  Node }
The UAVs under formation flight are in constant communication with their neighbors to compute their relative distance. In this scenario,  the $k^{th}$ UAV's outgoing signals are attacked by exogenous input or corrupted by a comunication network-induced deception attack while its control inputs are computed correctly \cite{kim2012cyber}. This scenario covers a DoS attack,  malicious data or noise injection to the UAVs connected to this $k^{th}$ UAV node. The attack can be modeled in the UAV dynamics as a sensor fault on the information broadcasted from this affected UAV node. Since this UAV is unaware of its outgoing information being corrupted, here the two measurements, internal measurement $(\vartheta_k)$ and the broadcast signal $(y_k)$, are isolated as stated in the system dynamics equation (\ref{eq:k node_out}) below

\begin{align}
\dot{x}_k&=A_kx_k+B_ku_k \label{eq:k node_out}\\
\vartheta_k&=C_kx_k \nonumber\\
y_k&=C_kx_k+C_f^kf_k \nonumber
\end{align}
where $ \vartheta_k\in\mathbb{R}^n$ is the internal measurement, $f_k\in\mathbb{R}^n$ being the corrupted broadcast information. The  closed loop dynamics can be written as
\begin{align}
\dot{x}&=(A+BKL)x-BKLh+I_{\bar{k}}\Gamma^kf_k \label{eq:k node_out_overall}\\
\vartheta&=Cx \nonumber\\
y&=Cx+\bar{C}_f^k f_k \nonumber
\end{align}
where  $y\in\mathbb{R}^{Nn \times N}$ is the communicated measurement and the corrupted feed to the network is $I_{\bar{k}}\Gamma^k$ where
$I_{\bar{k}}$  is obtained from identity matrix $I_{nN}$ with the $k^{th}$ diagonal element replaced with  block of $0_{n\times 1}$
and $\Gamma^k$ is the $k^{th}$ column of the $KLC$ matrix  to account for the internal measurement of $k^{th}$ UAV not being affected.
A UAV node $k$ will distinguish between an attack on the node itself and the outgoing communication based on  its internal
 measurement.

\begin{corollary}
	There exists a UIO for the system $((A+BKL),I_{\bar {k}} \Gamma^k,\bar{C}_i)$, if the graph $\mathcal{G}$ is connected.
\end{corollary}
\begin{proof}
	First note that $I_{\bar{k}} \Gamma^k$, is the $k^{th}$ column of  $KLC$ with $k^{th}$ entry set to zero, where both $K$ and $\bar{C_i}$ are full rank matrices. Therefore, if $\mathcal{G}$	is connected, node $k$ has at least one neighbor. Denoting the row of $\bar{C}_i$ that reads the output of node $k$ as $\bar{c}_i^k$ and rank of $(\bar{c}_i^kI_{\bar{k}}\Gamma^k )=$ rank $(I_{\bar{k}}\Gamma^k)=1$.
	
	The second condition of Theorem \ref{theo:2} is the detectability of a fault. It can be stated as: a fault is detectable if the transfer function of scalar $m$  faults $f_k(t)=[f_1(t),...f_m(t)] $ to $y(t)$ is not identical to zero, i.e. the rank of
$
		\begin{bmatrix}
			sI_{nN}-(A+BKL)&I_{\bar{k}}\Gamma^k\label{eq:com_attModel}\\
			\bar{C}_i&0\\
		\end{bmatrix}
$	
	$= nN+m$, for all $s$ such that $\mathfrak{Re}(s)\geq 0$. Here we deal with a single fault occurrence at a time, $m=1$.
		Rank	
$
\begin{bmatrix}
	sI_{nN}-(A+BKL)&I_{\bar{k}}\Gamma^k\\
	\bar{C}_i&0\\
\end{bmatrix}
$
$=$
$rank\begin{bmatrix}
	sI_{nN}-(A+BKL)\\
	\bar{C}_i\
\end{bmatrix}$
$+$
rank $ (I_{\bar{k}}\Gamma^k)= nN+1$
where a stable closed loop system $(A+BKL)$, is a full column rank. $ I_{\bar{k}}\Gamma^k$ matrix is a principal submatrix of the graph Laplacian. In \cite{barooah2006graph} it was shown that any principal matrix of a connected undirected graph Laplacian matrix is invertible and so the last column is independent of the rest.   \\
Therefore, rank	
$
\begin{bmatrix}
	sI_{nN}-(A+BKL)&I_{\bar{k}}\Gamma^k\\
	\bar{C}_i&0\\
\end{bmatrix}
$
$= nN+1$, is full column rank.	
\end{proof}
\subsection{UAV Under Attack Removal}
The  main role of FDI system here is to provide information about possible cyber attack or faults (detection) in the system  and determine the location of the fault or attack (isolation) to enable an appropriate reconfiguration to take place. Corrective action will be made to eliminate the threat or minimize the effect on the overall performance of the system \cite{alwi2011fault,tan2002sliding}.
In flight control system, it is important to determine the best appropriate control action following a system failure in order to ensure safe operation and continuity of service.  Once a faulty, a compromised or a malicious UAV is detected, a fault handling system would either depend on a fault tolerance controller or remove the UAV from the formation and maintain the system running possibly with graceful degradation of its performance. Here, the later alternative will be considered. Removing the compromised UAV node needs deleting the node followed by updating the communication graph and the control law in the consensus algorithm.
An algorithm developed on top of what was presented in our previous work \cite{kim2016cubature} with the necessary conditions to automatically remove the faulty/malicious UAV node.

To make sure removing a node in the graph will not create two or more disconnected graphs, the following assumption put forward. The graph $\mathcal{G}$ is assumed a 2-vertex-connected, i.e., after losing any single vertex it remains connected.
\begin{definition}
	A graph  $\mathcal{G}$ is 2-connected if $ \mid \mathcal{V}  (\mathcal{G})\mid \textgreater 2 $  and for every $x \in \mathcal{V} (\mathcal{G})$ the graph $ \mathcal{G}-x$ 	is connected.
\end{definition}
 \begin{theorem}
 A connected graph $\mathcal{G}$ with at least
 three vertices is 2-connected iff for every two vertices $x, y \in  \mathcal{V} (\mathcal{G})$, there is a cycle containing both.
 \end{theorem}
\begin{proof}
 (sufficient condition): If every two vertices belong to a cycle, no removal of one vertex can disconnect the graph.
	 (necessary condition): If  $\mathcal{G}$ is 2-connected, every two  vertices belong to a cycle.
	
	 Since a $2$-connected graph is also $2$-edge connected, i.e., after losing any single edge, it remains connected \cite{diestel2005graph}. The  graph model in the formation setup is cyclic from the arrangement and wireless broadcasting communication nature of the UAVs.

\end{proof}

The algorithm proposed described in Figure \ref{figAlg} is with the assumption that at most there is one compromised UAV with either of attacks at the neighborhood $i$. The algorithm will remove the faulty UAV node-$k$ from the network and updates the communication graph and control law.
	\begin{figure}[h]
	\centering
	\includegraphics[width=0.8\textwidth]{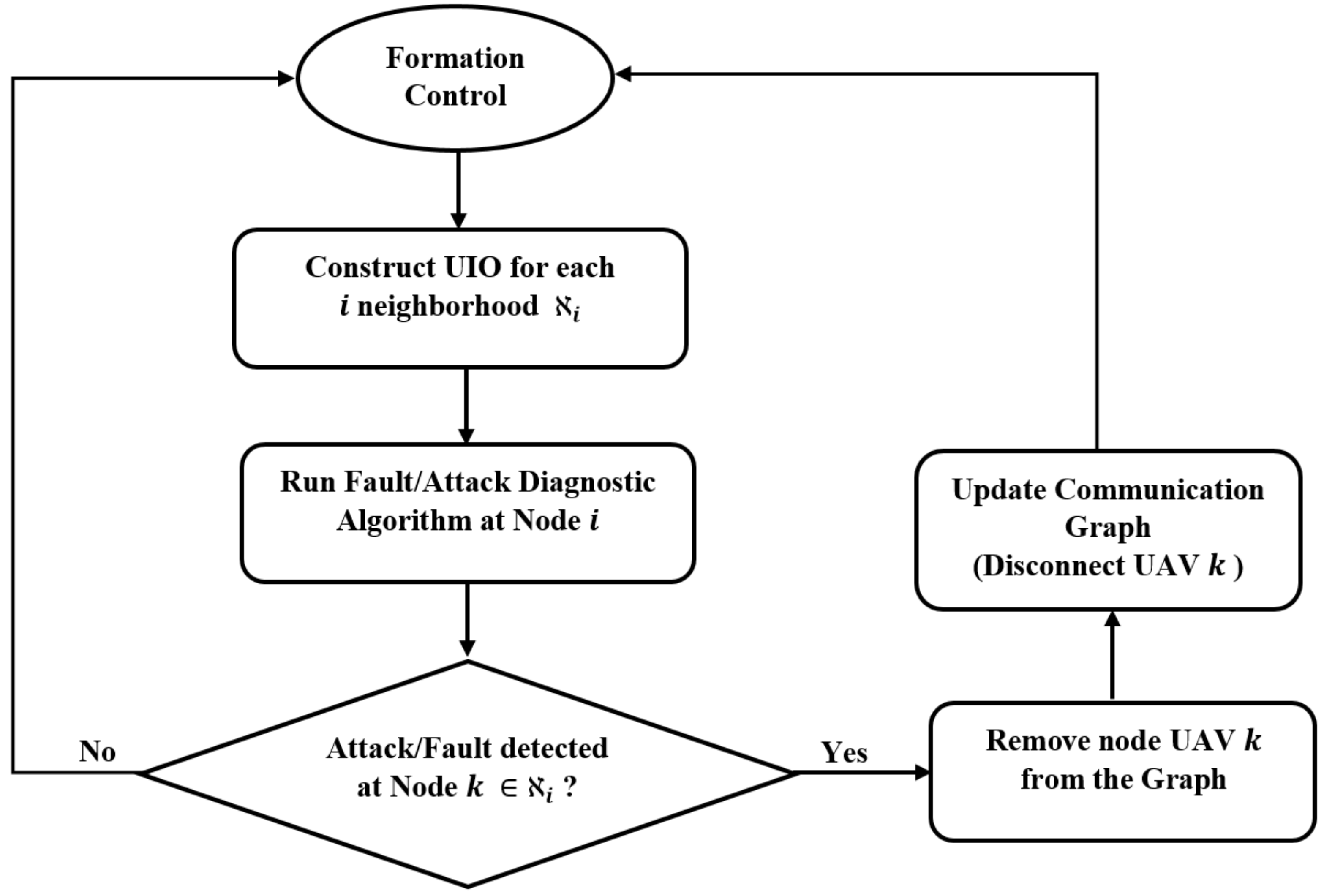}
	\caption{Compromised UAV Node removal and Formation control in the presence of an attack.}
	\label{figAlg}
\end{figure}

Each node identify the compromised UAV in the network using its bank of UIOs and update its communication graph accordingly. The graph theory property is exploited to remove the compromised UAV node. First, the adjacent matrix $\mathcal{A}$ is updated to $\mathcal{A} ^f$, where its element $a_{i,j}$ corresponding to faulty node $k$ is updated to $a^f_{i,j}$ as follows:
\begin{equation} \label{eq:adjUpdate}
a^f_{i,j} = 0, \quad \mathrm{if} \quad i=k \ \mathrm{or} \ j = k, \quad \mathrm{for \quad \forall}  i,j
\end{equation}

Second, the indegree matrix $\mathcal{D}$ is updated to  $\mathcal{D}^f$, where its diagonal element updated to $d^f_{i,i}$ as follows:
\begin{eqnarray}
d^f_{i,i} &=& d_{i,i} - 1 , \quad \mathrm{if} \quad i \neq k \quad \mathrm{for \quad  \forall}  i  \label{eq:eq_3_24} \\
d^f_{i,i} &=& 0, \quad \quad \quad \mathrm{if} \quad i=k \quad \mathrm{for \quad  \forall}  i \label{eq:eq_3_25}
\end{eqnarray}
Finally, the Laplacian matrix is updated and the control input gain is re-calculated using equation (\ref{eq:L_g}) and (\ref{OVR_allFormation}) respectively.
\section{Numerical Examples}
\subsection{Fault/Cyber Attack detection}
\par
To simulate  the two cyber attack types, a position offset and a communication deception are introduced in one of the UAVs in the formation setup. A bank of UIO based fault detection scheme is implemented on each of the UAVs, while each monitoring its neighbors to detect a cyber attack and identify the faulty or compromised UAV in the network which is connected according to the communication graph. For a numerical simulation purpose, an attack only on UAV-2, and the distributed attack/fault detection bank of UIOs on UAV-1 is considered. First, a fault/attack free formation flight of $6$ UAVs presented to see if the consensus based distributed controller managed to keep a specified hexagon formation flight of the UAVs.
\subsubsection{A Fault/Attack Free Formation Flight }

To illustrate the formation control, six UAVs at a hexagon corner considered with a formation vector $\tilde{h}=[(2~ , 0)^T ~ (1 ~,-1.73)^T~ (-1~ ,-1.73)^T~ (-2~, 0)^T~ (-1 ~,1.73)^T~ (1~ ,1.73)^T]$, hexagon centered at $(0,0)$ and radius of $2$ meters. The distributed control maintains a hexagon formation as illustrated in Figure \ref{fig3}. The UAVs considered to be in a level flight and started from arbitrary $x,y$ positions.
\begin{figure}[htbp]
	\centering
	\includegraphics[width=0.45\textwidth]{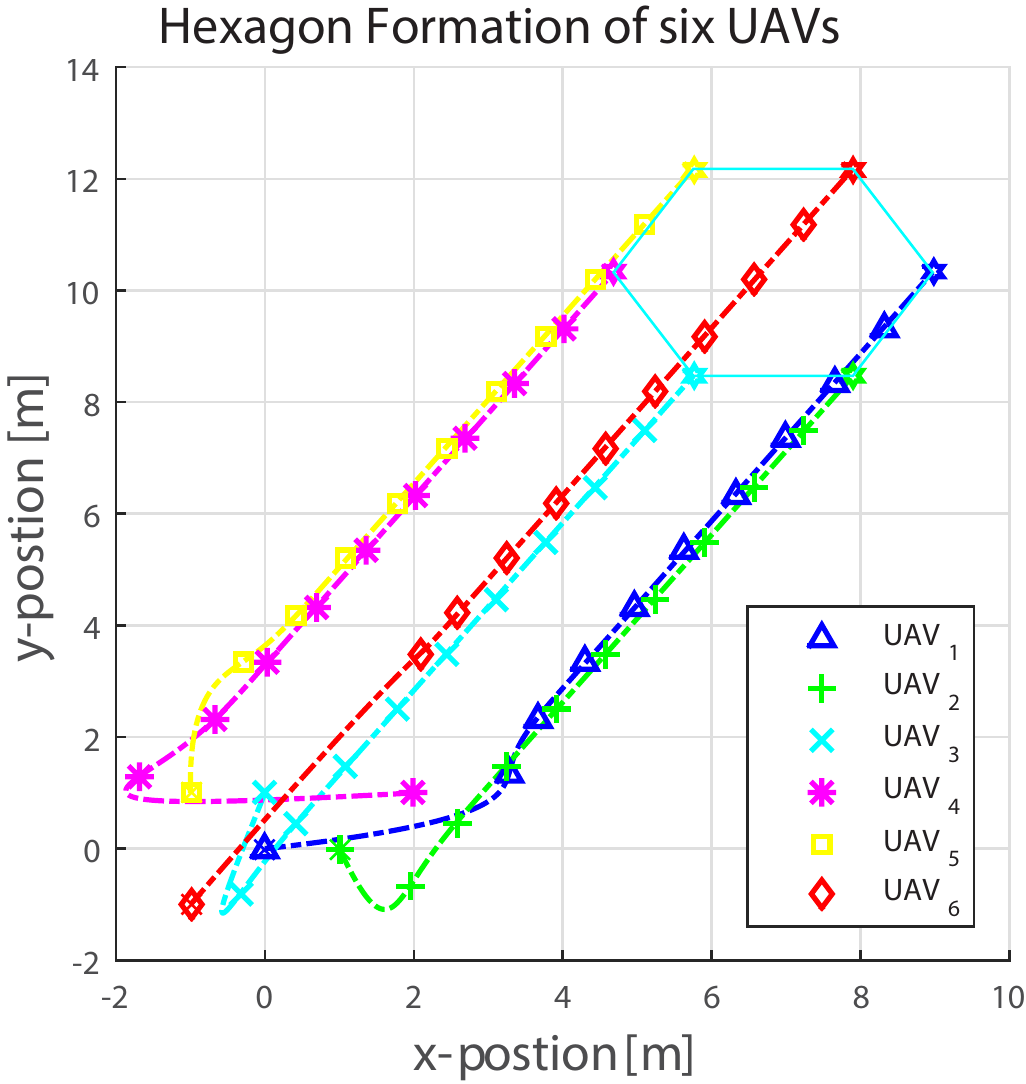}
	\caption{Six UAVs in Hexagon formation.}
	\label{fig3}
\end{figure}
\subsubsection{Node Attack}
UAV-2 is suffering from a node attack, modeled as an offset in its position in the formation setup (\ref{eq:k nodeAtt dym}). The node attack, offset in the UAV-2's x-position is introduced in the time interval between 0.5 and 4 seconds. It is evident from Figure \ref{fig4} that the hexagon formation is no more in place for the
\begin{figure}[htbp]
	\centering
	\includegraphics[width=0.5\textwidth]{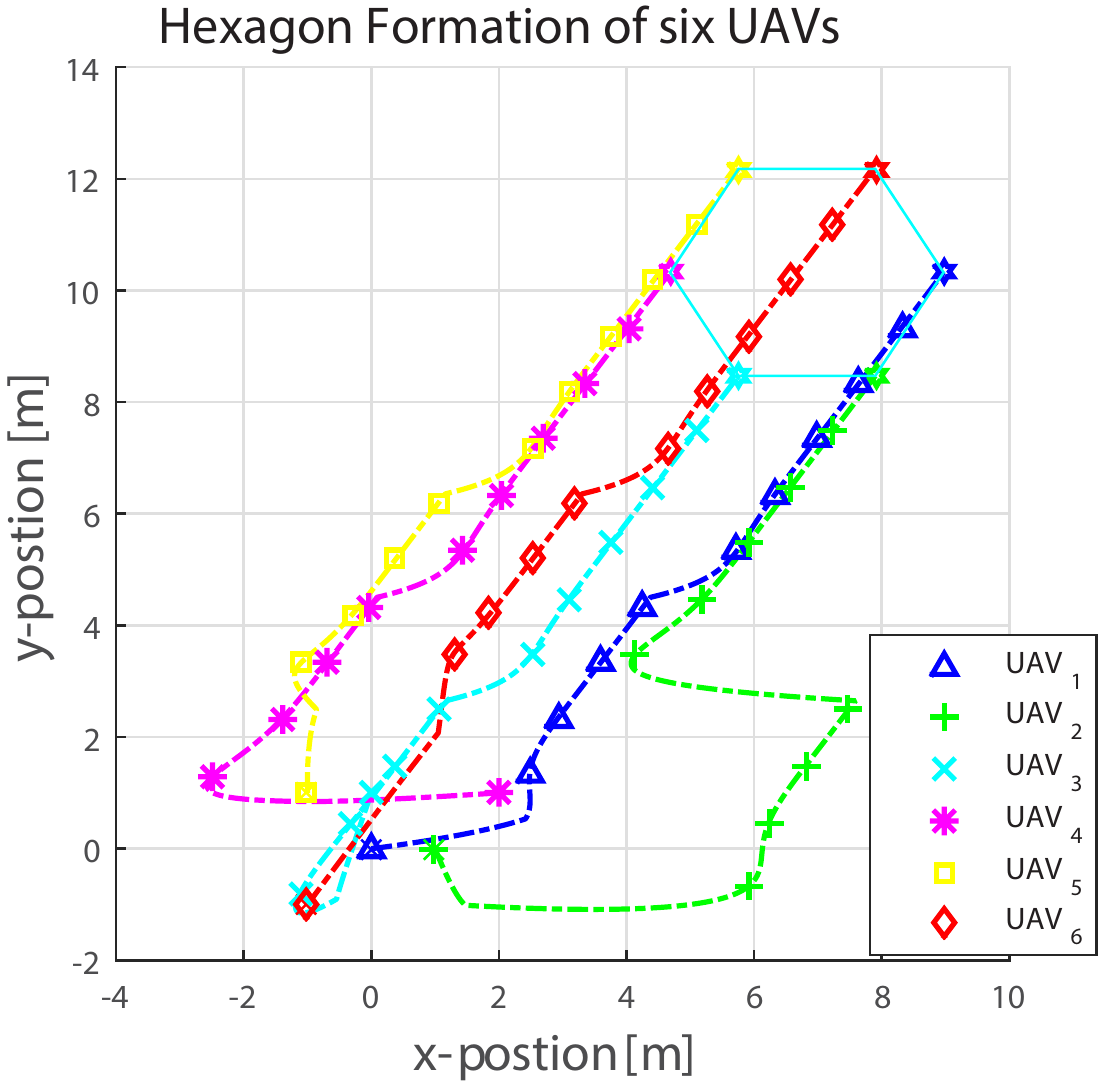}
	\caption{Formation UAV-2 under node attack in the time interval 0.5 and 4 seconds.}
	\label{fig4}
	\end{figure}
	\begin{figure}
	\begin{subfigmatrix}{2}
		\subfigure[x-positon error]{\includegraphics{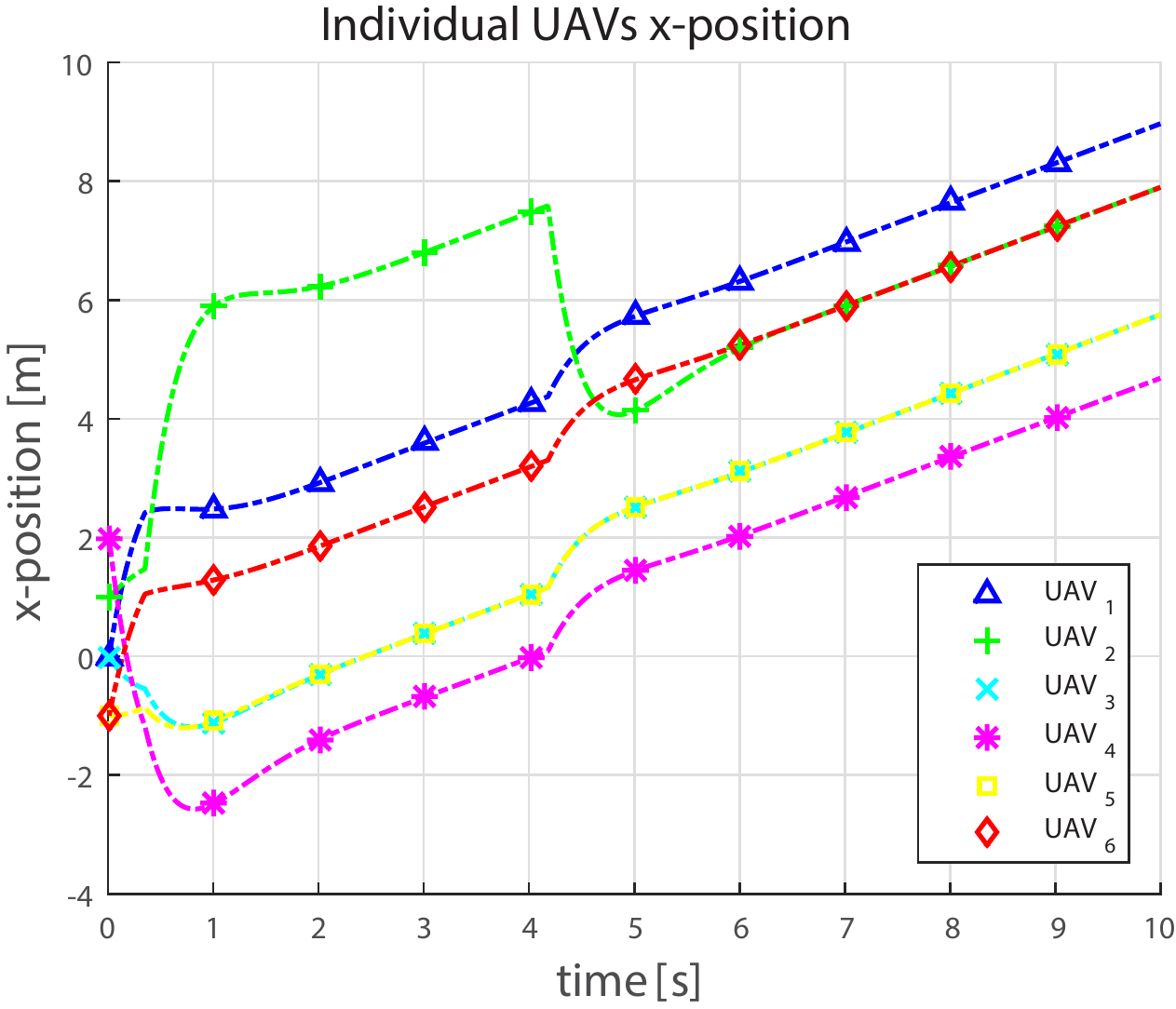}\label{fig5a}}
		\subfigure[Residual]{\includegraphics{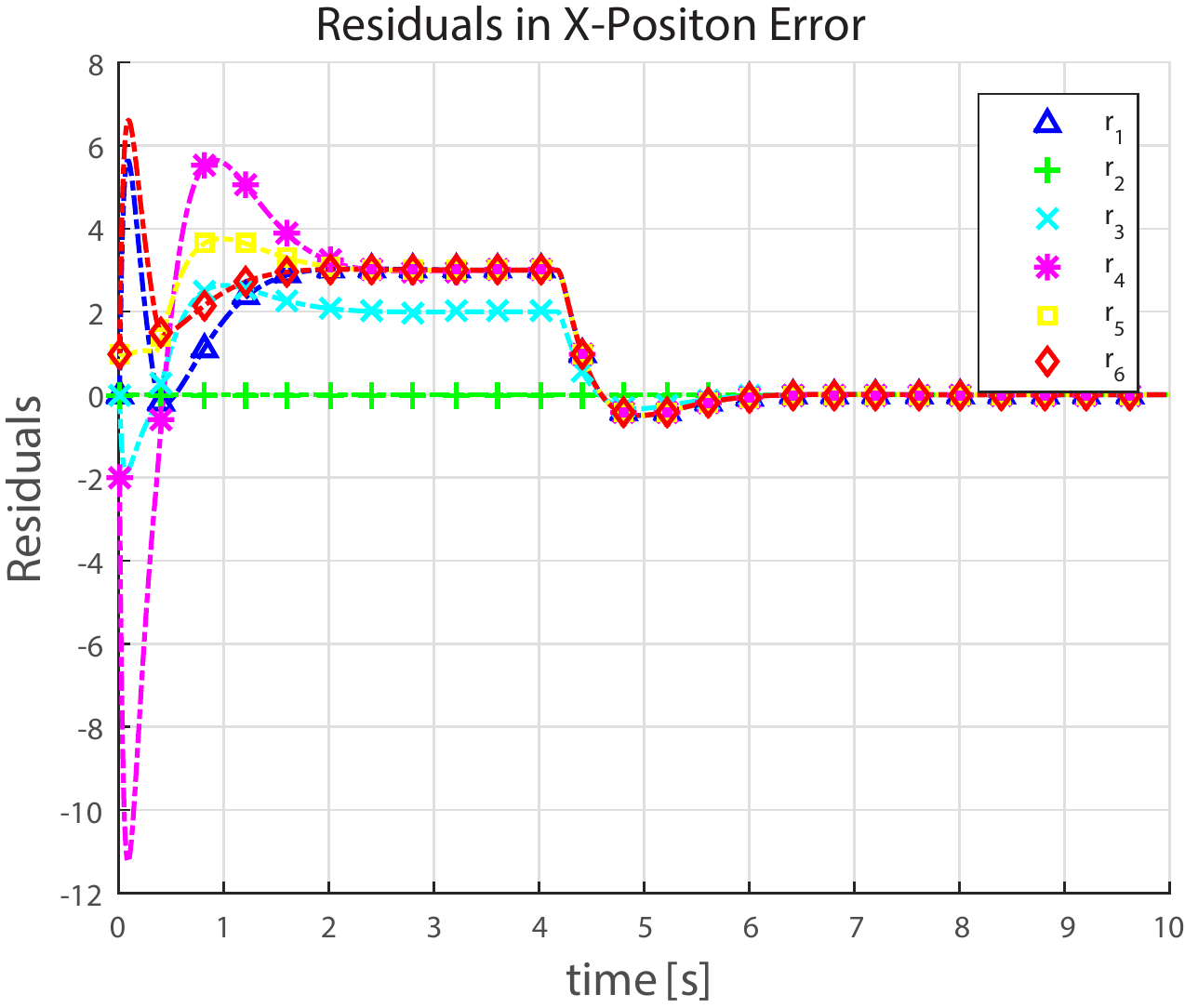}\label{fig5b}}
	\end{subfigmatrix}
	\caption{(a) x-position error introduced in the formation due to the node attack in the time mark between 0.5 and 4 seconds.
		(b) Residual generated at UAV-1.}
	\label{fig5}
\end{figure}
specified time period not only for the affected UAV but also the others.
This is because the consensus algorithm uses the relative position of the UAVs in the formation to calculate the feedback gain, hence the others will be affected too. The node attack on UAV-2 effect is seen in the Figure \ref{fig5a}, where the UAVs are no more able to maintain their x-position to complete the hexagon formation. The effect of the attack as an offset is more pronounced on a directly compromised UAV-2, while on its neighbors much subdued as it is affecting them through a feedback. UAV-2 also reacted differently since it had the disturbance, modeling the attack, in its own dynamics as it is  given on the fault scenario (\ref{eq:k nodeAtt dym}). A bank of UIOs' residuals at UAV-1 for each of its neighbor is plotted in Figure \ref{fig5b}. As it can be seen, from Figure \ref{fig5}(b), the residual corresponding to UAV-2 is zero while all other are larger, i.e. the UIOs at UAV-1 is made by design insensitive to a fault introduced by UAV-2 in the network of the UAVs. Using the threshold logic presented in section III.B.1, the agent UAV-1 not only detects the cyber attack in the network but also identifies UAV-2 as being a compromised
node in the formation.

\subsubsection{Out going Communication Attack}
\par
Much of the vulnerability of UAVs to a cyber attack is due to their communication to their surrounding environment. To illustrate a communication induced deception attack in the network of UAVs under formation flight, two form of attacks are presented  with different capability of an adversary in question: Offset introduced and Noise injection in the communication channel. A UAV under out going communication attack can distinguish between an attack on the node itself and the outgoing communication based on its internal measurement.
\paragraph{  Offset Introduced :}
Assume an adversary manages to get access to the broadcasting module of the UAV-2 in the network and introduced a bias in the outgoing broadcasted  signal of UAV-2 as explained in section III.B.2 (\ref{eq:k node_out}).
\begin{figure}[h]
	\centering
	\includegraphics[width=0.45\textwidth]{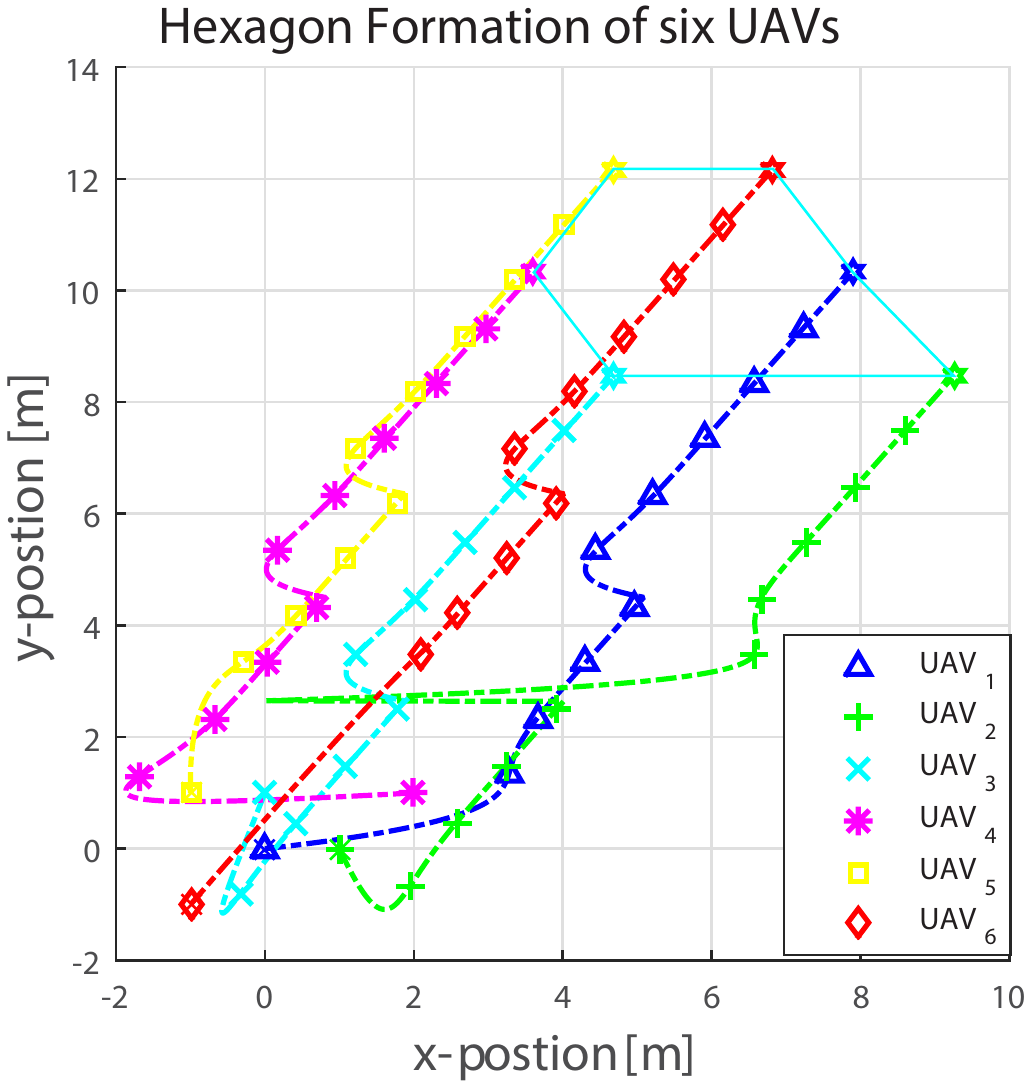}
	\caption{ Hexagonal formation UAV-2 under offset introduced communication induced deception attack in the time interval 0.5 and 4 seconds.}
	\label{fig6a}
\end{figure}
 The bias introduced  at a time mark of 4 seconds into the flight. UAV-2 is not aware of its outgoing signal is being corrupted while still computing its own control signals correctly.
   \begin{figure}[!htbp]
 	\begin{subfigmatrix}{2}
 		\subfigure[Hexagon formation $x-y$ position ]{\includegraphics{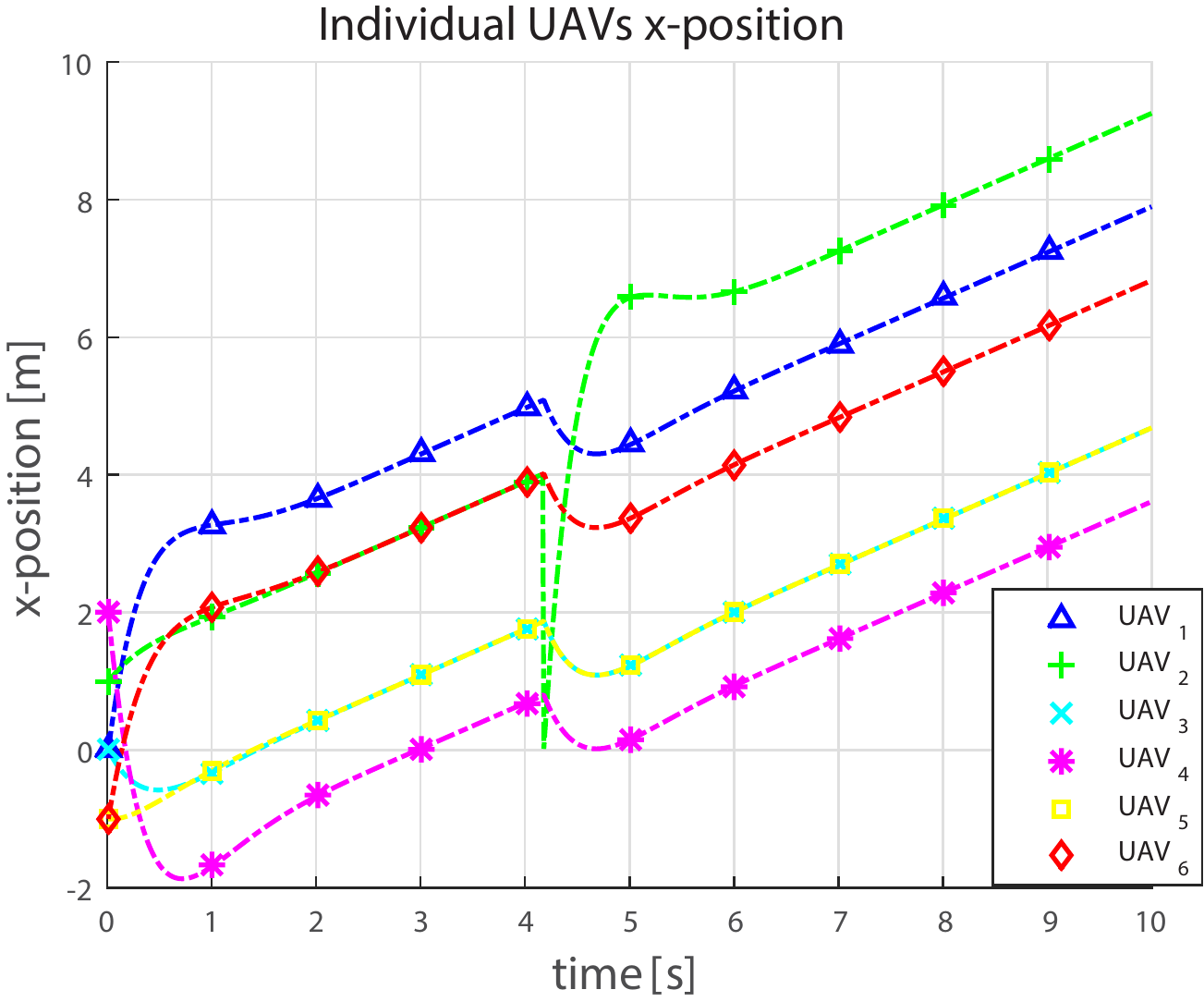}\label{fig6b}}
 		\subfigure[Residual]{\includegraphics{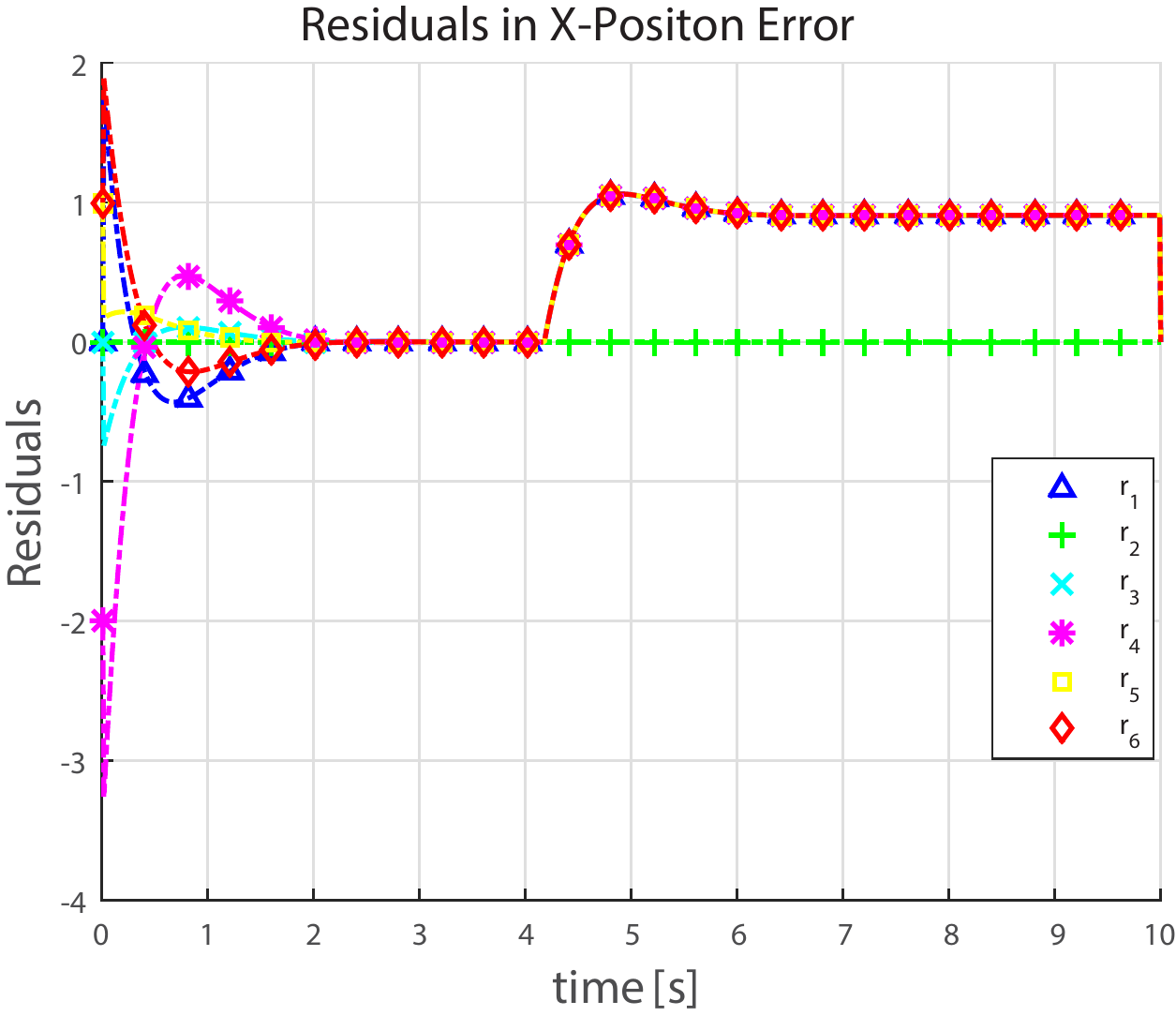}\label{fig6c}}
 	\end{subfigmatrix}
 	\caption{(a) Hexagonal formation, UAV-2 under offset introduced communication deception attack
 		(b) Residual generated at UAV-1.}
 \end{figure}

As it can be seen in Figure \ref{fig6a}, the other five UAVs except the one maliciously sending a corrupted signal to the rest, are in a hexagon formation. As they trust the signal coming from the UAV-2, they align themselves or reach consensus on the wrong measurement to the hexagon formation. UAV-2 is in offset, since it calculates its control signal from the uncorrupted measurement it has taken but still follows the rest of the UAVs as it is running the consensus algorithm. While building a bank of UIOs at each UAV, the concept behind is that each UAV should be able to check if it is behaving correctly using its internal measurement and communicated signals from its neighbors. As compromised UAV is not aware of its own transmitted data being corrupted and assumes all neighbor UAVs are misbehaving. As illustrated in Figure \ref{fig6c}, a bank UIOs based residuals generated at UAV-1 detects a cyber attack and also successfully identifies UAV-2 as the compromised node using the threshold logic. After the FDI system detects and isolates UAV-2 is compromised, at this point UAV-2 can  redefine its communication security key to avoid itself being disconnected from the rest of the formation flying UAVs.

\paragraph{Noise Injected in the Communication Path:}

Here, it is assumed that the adversary knows the communication channel model. This enables it to corrupt the message shared by UAV-2 with a random signal or inject a false data into the signals being communicated.
\begin{figure}[!htbp]
	\centering
	\includegraphics[width=0.45\textwidth]{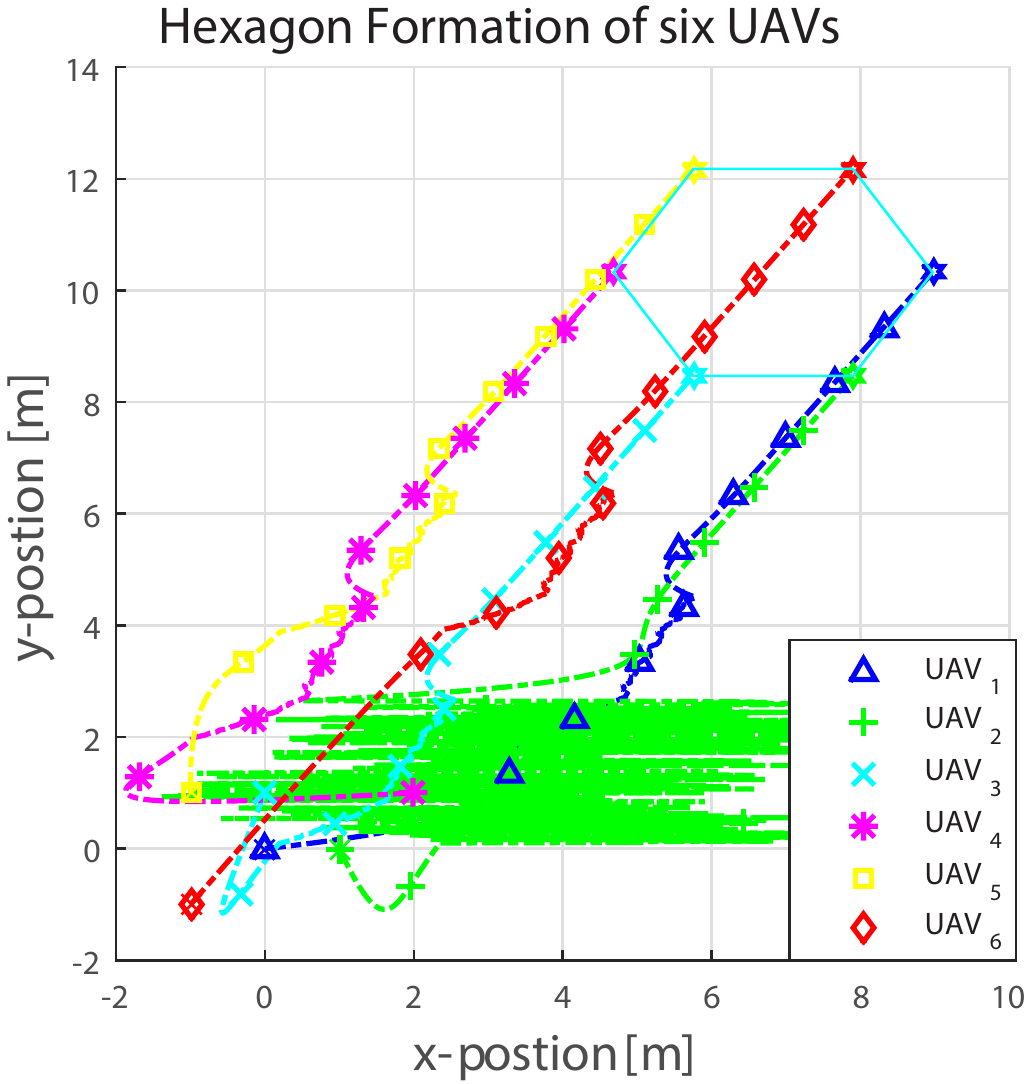}
	\caption{Formation under UAV-2 node random data injection attack at a time mark between 2 and 5 seconds.}
	\label{fig7Ja}
	\begin{subfigmatrix}{2}
		\subfigure[x-positon error]{\includegraphics{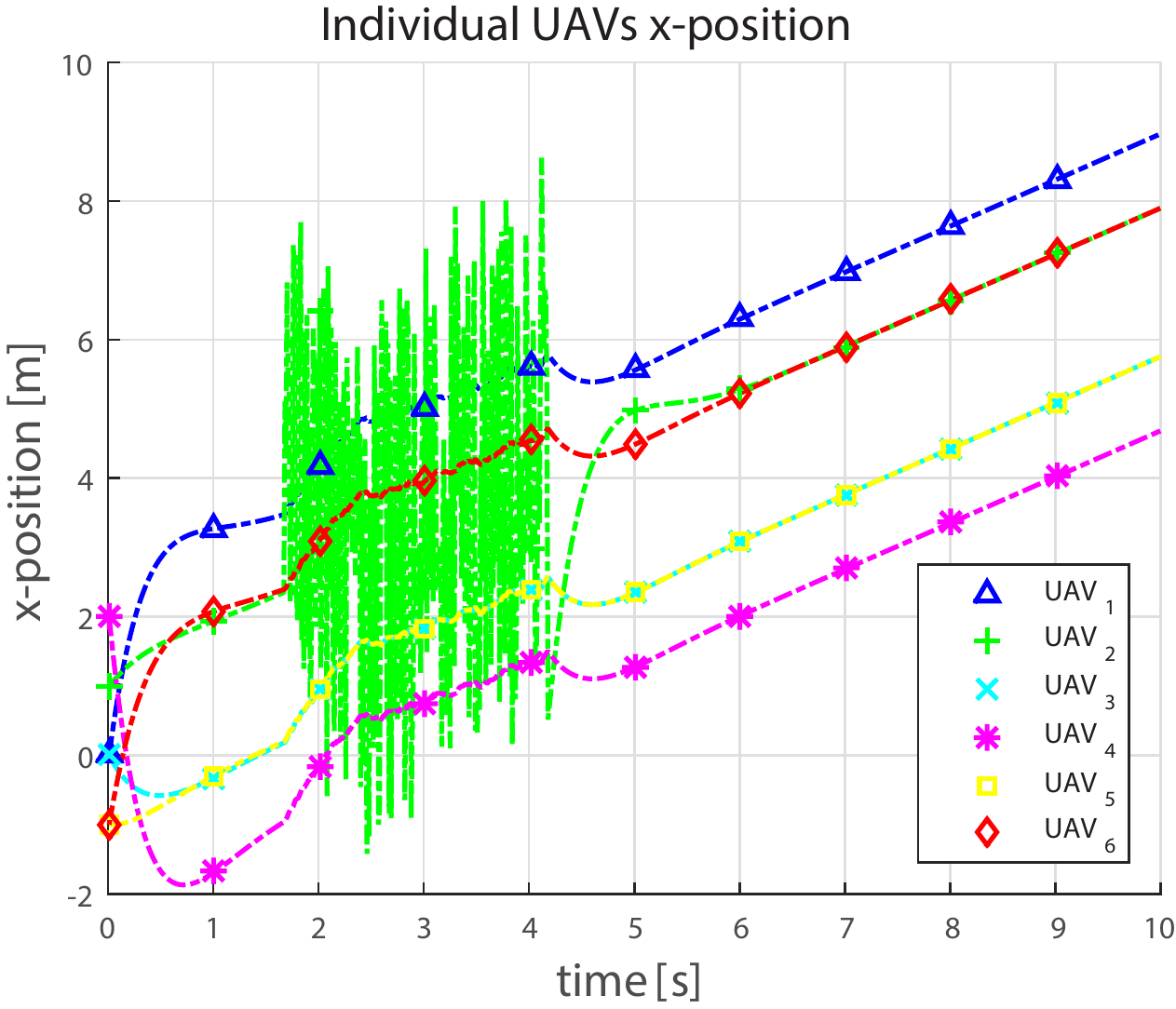}\label{fig8a}}
		\subfigure[Residual]{\includegraphics{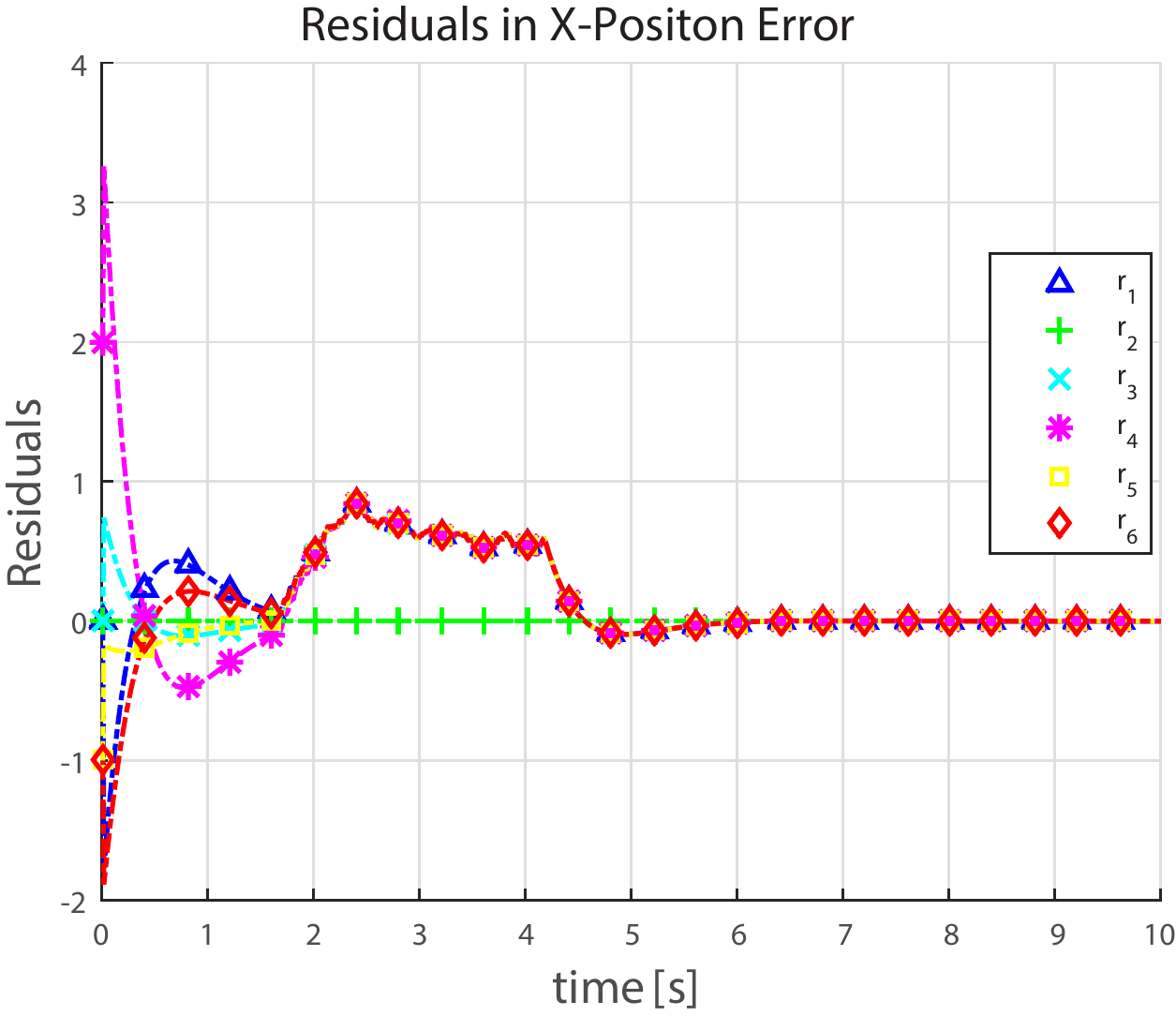}\label{fig8b}}
	\end{subfigmatrix}
	\caption{(a) x-position error introduced in the formation due to the UAV-2 node random data injection attack at a time mark between 2 and 5 seconds.
		(b) Residual generated at UAV-1.}
\end{figure}
With the same argument as in case of $'$Offset Introduced$'$ above, the signal received by the rest of the UAVs is being corrupted as it passed through  the vulnerable  communication channel in that attack time interval. As it can be seen in in Figure \ref{fig7Ja}, the consensus algorithm is unable to keep the formation in place, especially the one being affected by the noise, UAV-2. As illustrated in Figure \ref{fig8b}, a bank of UIOs based residual generated at UAV-1, all the residuals except residual from UAV-2 is non-zero, indicating the source of attack is UAV-$2$.

\subsection{UAV Under Attack Removal}
In a formation flight when one of the UAV misbehaves, either because of the above mentioned cyber attack or a fault, the formation flight will be no more in place. An attack introduced early, 4 seconds in the hexagon formation flight. The hexagon formation flight disrupted (Figure \ref{fig10}) as the compromised UAV-2 behaves differently due to the disturbance in its own dynamics representing the effect of the cyber attack.
\begin{figure}[!htbp]
	\begin{subfigmatrix}{2}
				\subfigure[Residual]{\includegraphics{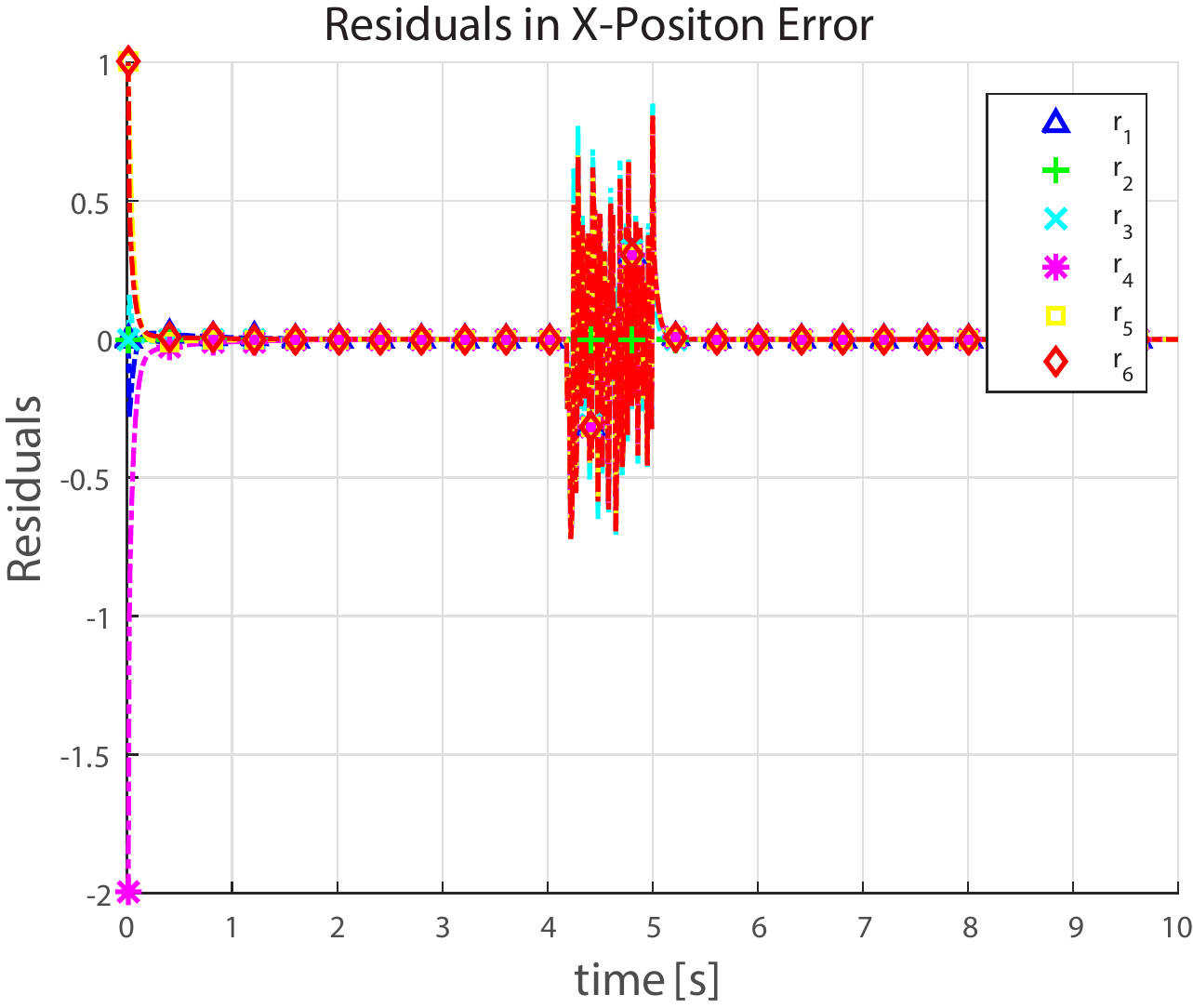}\label{fig9a}}
		\subfigure[x-positon error]{\includegraphics{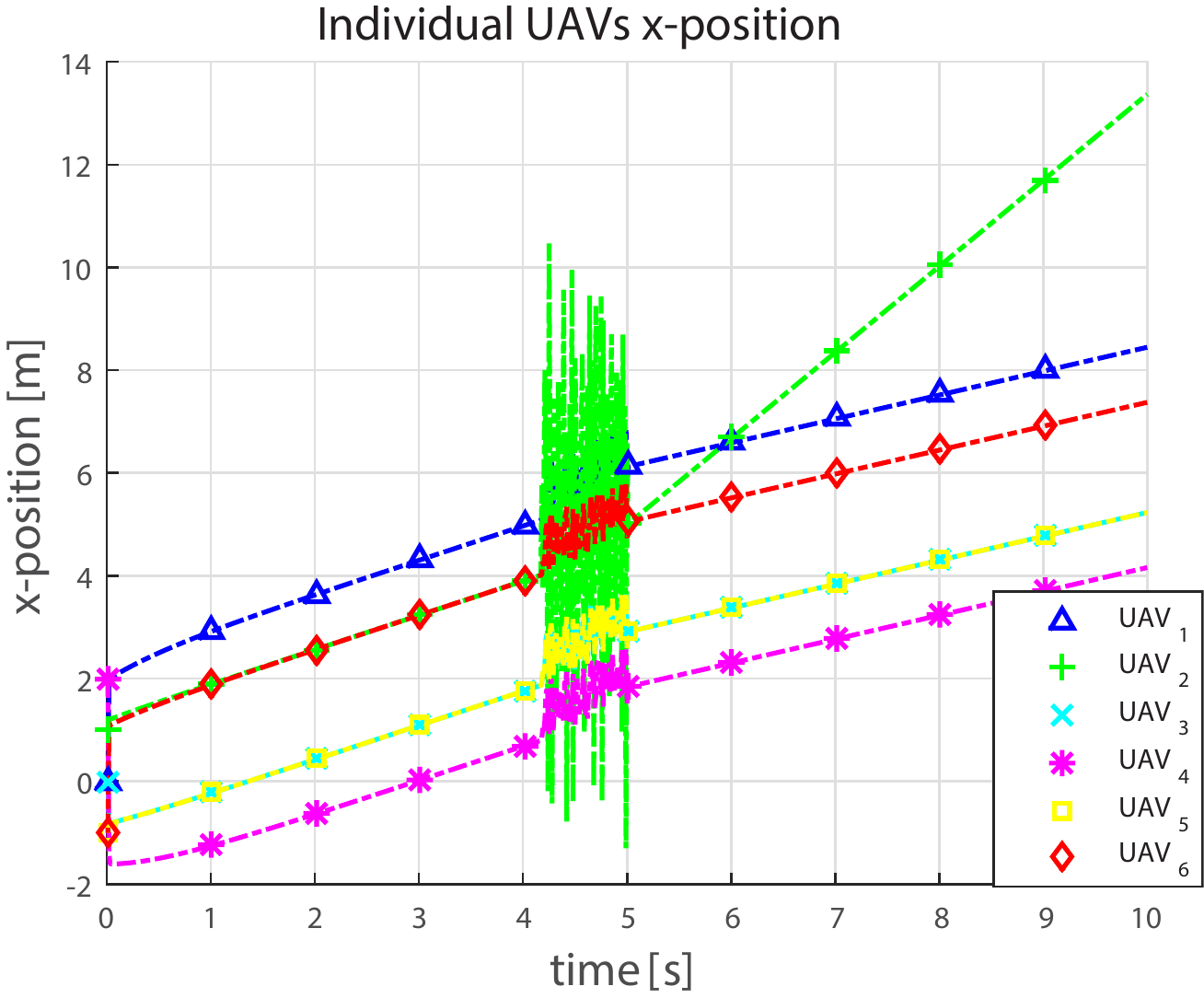}\label{fig9b}}
	\end{subfigmatrix}
	\caption{(a) Residual generated at UAV-1. (b) x-position random offset error introduced in the formation at node UAV-2 at a time mark  4 seconds into the flight.}
\end{figure}
\begin{figure}[!htbp]
	\centering
	\includegraphics[width=0.45\textwidth]{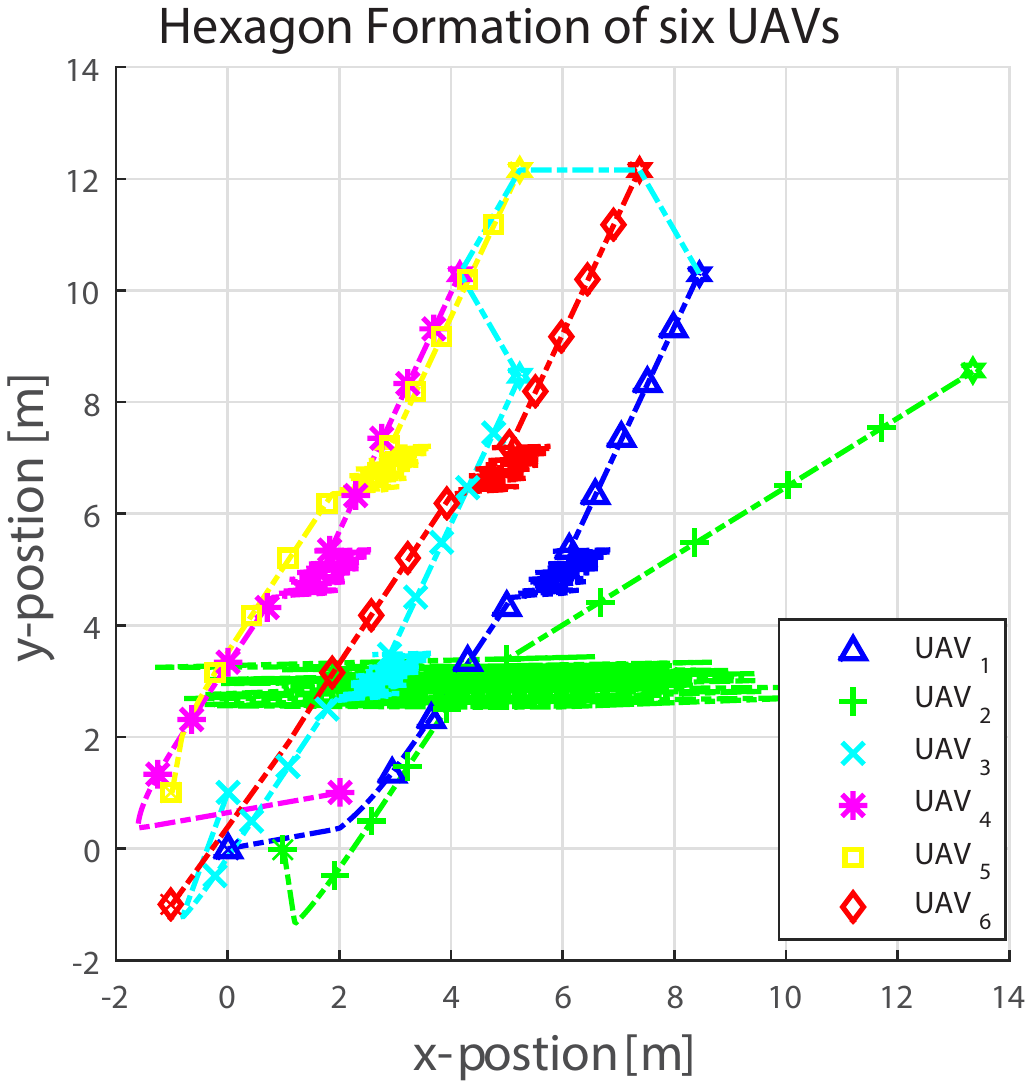}
	\caption{Faulty or Under attack UAV removal in Hexagon formation flight.}
	\label{fig10}
\end{figure}
As it could be seen for a short instant in Figure \ref{fig9a}, the FDI system detect and identified the malicious or compromised UAV as UAV-2. The compromised UAV made the over all formation flight short of any use. Once an attack detected in the formation flight, faulty node UAV removal algorithm  invoked. With the assumption forwarded, the connection graph is 2-connected, removing the compromised UAV will not create two separate network of UAVs which can't communicate to each other to reach consensus on their formation flight variables. The algorithm results in the  removal of the UAV-2 as it can be seen in Figure \ref{fig10}. The formation flight is kept in place with missing corner of the hexagon formation at node $2$. If the formation flight was meant for find and rescue mission or sensor networks, it would serve the purpose with degraded performance than losing the whole purpose of the flight.
	\section{Conclusions}
In this paper, detection of cyber attacks has been considered on a network of UAVs in formation flight. Because of the nature of formation flight and the control algorithm used, a distributed fault detection and isolation scheme proposed based on a bank of an unknown input observers which only requires local measurements. The proposed  fault detection scheme not only able to detect a cyber attack but also successfully identified the compromised UAV in the formation network. Furthermore, an algorithm  has been proposed to safely and automatically remove the faulty UAV or a UAV under attack while keeping the formation with degraded but functioning performance. Finally, a numerical case study have been given with a typical example of six UAVs in a hexagon formation with a possible node and communication deception attacks. Finally, a numerical case study has demonstrated that the residual generated at the monitoring node UAV able to successfully detect and isolate the cyber attack. Also, the faulty UAV removal algorithm has been shown effectively remove the  compromised UAV to maintain the formation accordingly. Future work includes extension of the proposed scheme to handle more complex attack patterns and applying the method for other types of multi-UAV coordination missions.

	\section{Acknowledgments}
	
	This work was supported in part by ICT R\&D program of MSIP/IITP [R-20150223-000167, Development of High Reliable Communications and Security SW for Various Unmanned Vehicles].


\bibliography{BibFile}
\bibliographystyle{aiaa}

\end{document}